\documentclass{cccg17}
\usepackage{graphicx,amssymb,amsmath,jeffe}

\newcommand{\dist}{\operatorname{d}}
\newcommand{\Reeb}{\mathcal{R}}

\title{Supporting Ruled Polygons}

\author{Nicholas J. Cavanna\thanks{University of Connecticut, {\tt nicholas.j.cavanna@uconn.edu}}
        \and
        Marc Khoury\thanks{University of California, Berkeley, {\tt khoury@cs.berkeley.edu}}
        \and
        Donald R. Sheehy\thanks{University of Connecticut, {\tt don.r.sheehy@gmail.com }}}

\index{Cavanna, Nicholas}
\index{Khoury, Marc}
\index{Sheehy, Donald}

\begin{document}
\thispagestyle{empty}
\maketitle

\begin{abstract}
We explore several problems related to ruled polygons. Given a ruling of a polygon $P$, we consider the Reeb graph of $P$ induced by the ruling. We define the Reeb complexity of $P$, which roughly equates to the minimum number of points necessary to support $P$. We give asymptotically tight bounds on the Reeb complexity that are also tight up to a small additive constant. When restricted to the set of parallel rulings, we show that the Reeb complexity can be computed in polynomial time.
\end{abstract}

\section{Introduction}
\label{sec:intro}

Gauss's Theorema Egregium states that any isometric embedding of a surface preserves the (Gaussian) curvature everywhere on the surface~\cite{Pressley2010}.
A particularly important example of this is the case of flat, or \emph{rectifiable}, surfaces.
The Gaussian curvature is the product of the so-called principal curvatures, and so, zero-curvature implies that at every point, some direction lies in a straight line (a principal curvature of zero).
Ruled surfaces are one such example, but the most well-known example is pizza.
If one rolls a (flat) triangular piece of pizza in one direction, the curvature in that direction is non-zero and thus, unless the pizza stretches or tears, the curvature in the orthogonal direction must be zero.
This is what keeps the tip of the pizza from flopping downward.

In this paper, we attempt to establish a theoretical foundation for algorithmic problems on isometric embeddings of rectifiable polygons in $\R^3$.
These are planar polygons that have embedded in $\R^3$ so that the curvature stays zero everywhere locally, but the embedding may not lie in a plane.
Our original motivation came from the question of how many single points of contact where necessary to support a polygon so that none of the corners can ``flop''.
A similar problem was studied in robotics for holding cloth~\cite{Bell2010}, in which the researchers rediscovered the Art Gallery problem~\cite{ORourke1987}.
To correctly state such problems in the zero-curvature case, we first establish a vocabulary for \emph{rulings} and give a description of the intrinsic topology.
The ruling is the set of lines of zero curvature in some isometric embedding.
This allows us to abstract away issues of physics (gravity, for example) and embeddings (self-intersection).
Along the way, we connect these rulings to a generalization of Reeb graphs that allows us to connect ruled polygons to an art gallery-type theorem, phrased in terms of topological simplification.

For a ruled polygon $P$ with no holes, cutting along a ruling line divides it into two pieces.
The ruling is \emph{supported} by a set of points $S\subset P$ if every line $\ell$ of the ruling has a point of $S$ on both pieces of $P\setminus \ell$.
For example, two points suffice to support a triangle (slice of pizza).
We give the formal definition of the \emph{Reeb complexity} of a polygon in Section~\ref{sec:bounds}, but roughly, it corresponds to the minimum size support set over all possible rulings.
In Section~\ref{sec:bounds}, we prove that all $n$-gons have Reeb complexity at most $\frac{n}{2}+1$, and we give  a family of polygons with Reeb complexity $\frac{n}{2} - 4$.

We conclude with a collection of open problems and research directions that we believe could be of interest as they provide connections between classic problems in computational geometry such as monotone polygons, Hamiltonian triangulations, and art galleries and growing new areas such as Reeb graphs and topological simplification.
\section{Definitions}
\label{sec:definitions}

Let $P$ be a simple polygon in the plane.
Let $\hat{P}$ be an isometric embedding of $P$ into $\R^3$.
This is one for which the distance between any pair of points on the embedded surface is the same as in the plane.
Every point on $\hat{P}$ will have a principle curvature of zero in some direction.
We will limit ourselves to \emph{nondegenerate} embeddings, in which there is a unique such direction at each point.
These correspond to line segments covering the polygon.
Rather than working directly with embeddings $\hat{P}$ of $P$, we will look at the patterns induced by these line segments on $P$ itself.
Thus, we use Gauss's theorem to make statements about nondegenerate isometric embeddings by reasoning directly about planar polygons.

A \emph{ruling} of $P$ is a set of line segments in $P$ with both endpoints on the boundary, whose interiors partition the interior of $P$.
Moreover, we require that no two distinct segments in a ruling are collinear and intersect.
This last condition may seem strange at first, but is a fundamental issue in rulings, particularly in defining a topology on the rulings.
According to this definition, the segment through the reflex vertex of the polygon cannot be replaced by two shorter line segments.
Segments that contain a reflex vertex in their relative interior are called \emph{branch segments}.

A degenerate embedding of $P$ does not have a corresponding ruling.
This happens both for the extreme cases where the entire embedding lies in a plane, but also in other more interesting cases.

We say that a ruling is \emph{simple} if it has no branch segments.
A ruling is \emph{parallel} if every pair of segments are parallel.
A ruling is \emph{Morse} if no branch segment contains more than one reflex vertex in its relative interior.
\section{The Topology of Rulings}
\label{sec:topology}
Given a simple polygon $P\subset \R^2$ with or without holes, the \emph{Reeb graph}~\cite{Reeb1946} with respect to a continuous function $f:P\to \R$ is the quotient space $\Reeb(f):= P/\sim_f$ where $x\sim_f y$ if and only if $x$ and $y$ are in the same connected component of $f^{-1}(c)$, where $f(x)=f(y)=c$. For an accessible introduction to Reeb graphs, see~\cite{Edelsbrunner2010}.
We can construct a similar space from the ruling lines of $P$.
That is we define the \emph{Reeb graph of a ruling} $S$ of $P$, $\Reeb(S)$, to be the quotient space $P/\sim$, where $x\sim y$ if and only if $x,y\in s$ for some segment $s\in S$.
The branch segments of $S$ decompose the polygons into pieces which correspond to edges in the Reeb graph, glued together at internal nodes corresponding to the branch segments themselves.
$\Reeb(S)$ has a natural graph metric (i.e., is a 1-dimensional stratified metric space) where the distance between two equivalence classes $[x],[y]$ is the Hausdorff distance between the line segments containing $x$ and $y$.
Recall that the Hausdorff distance between two compact subsets $A$ and $B$ of a metric space is defined as
\[
  \dist_H(A,B) = \max\{\max_{a\in A}\min_{b\in B}\dist(a,b), \max_{b\in B}\min_{a\in A}\dist(a,b)\}
\]
The constructions of a Reeb graph of a ruling and the Reeb graph constructions align for particularly nice rulings, hence the naming convention.
If we have a parallel ruling $S$, then the Reeb graph of the ruling is equivalent to the Reeb graph formed from the height function orthogonal to the ruling.
Alternatively, if we have a simple ruling $S$ on $P$, then we may consider the midpoint $m_s$ of each line segment $s\in S$, which collectively trace out a path $\gamma: [0,1]\rightarrow P$.
We can then define a continuous function $f: P\rightarrow \R_{\geq 0}$ by $f(p)=\int_0^{t_0} |\gamma'(t)| dt$, where if $p$ is on line segment $s$, $f^{-1}(m_s)=t_0$,  yielding $\Reeb(S)=\Reeb(f)$.

For completeness' sake, the remainder of this section will explore the relations between the homotopy type and homology of a polygon $P$ and its Reeb graph of a ruling $S$, $\Reeb(S)$.
A known result in~\cite{Calcut904quotientmaps} states that if a space $X$ is locally path connected and is partitioned into connected equivalence classes by $\sim$ and $X/\sim$ is semilocally simply connected, then $q_*: \pi_1(X)\rightarrow \pi_1(X/\sim)$ is surjective, where $q$ is the topological quotient map.
Since each $q^{-1}([p])$ are path-connected, and all other conditions are satisfied, we have that $q_*: \pi_1(P)\rightarrow \pi_1(\Reeb(S))$ is a surjection. Note that we use the fact that $\text{int}(P)$ is homotopy equivalent to $P$.

With regards to homotopy, if $P$ has $h$ holes then it is homotopy equivalent to $\bigvee_h S^1$, where $\bigvee$ is the wedge sum formed by taking the disjoint union of $h$ copies of $S^1$ and adjoining them each at a single point.
The fundamental group of $P$, $\pi_1(P)$, is then equal to $\Z\ast \ldots \ast \Z$, the free product on $h$ generators.
If $P$ has no holes, then it is contractible, and we have that $P$ and $\Reeb(S)$ are homotopy equivalent.

With regards to homology, the Hurewicz Theorem states there is an isomorphism between $\pi_1^{\text{ab}}(P)$ and $H_1(P)$, where the former is the abelianization of $\pi_1(\cdot)$.
Since $(\ast_{i=1}^h \mathbb{Z})^{\text{ab}}= \mathbb{Z}^h$, then $H_1(P)=\mathbb{Z}^h$ as expected.

The following theorem due to Vietoris~\cite{Vietoris1927} allows us to provide an isomorphism between the homology of $P$ and that of $\Reeb(P)$.

\begin{theorem}[Vietoris-Begle Mapping Theorem]\label{thm:vietoris_thm}
Given compact metric spaces $X$ and $Y$ and surjective map $f:X\rightarrow Y$,  if for all $0\leq k\leq n-1$, for all $y\in Y$, $\tilde{H}_k(f^{-1}(y))=0$, then $f_*: \tilde{H}_k(X)\rightarrow \tilde{H}_k(Y)$ is an isomorphism for $k\leq n-1$ and a surjection for $k=n$.
\end{theorem}

Consider the quotient map $q: P\rightarrow \Reeb(S)$, which is surjective by definition. Given $[p]\in\Reeb(S)$, each fiber $q^{-1}([p])$ is the line segment corresponding to the equivalence class $[p]$, thus contractible, so $\tilde{H}_k(q^{-1}([p]))$ is acyclic for all dimensions $k$. Theorem~\ref{thm:vietoris_thm} then implies that $\tilde{H}_*(P)=\tilde{H}_*(\Reeb(S))$.
\section{Asymptotically Tight Bounds}
\label{sec:bounds}
Let $P$ be a simple polygon with $n$ vertices and $h$ holes. We define the $\emph{Reeb complexity}$ of $P$ as the minimum number of leaves in $\mathcal{R}(S)$ over all possible rulings $S$ of $P$. In this section we show that the Reeb complexity of $P$ is upper bounded by $\frac{n}{2}+1$ and can be as large as $\frac{n}{2} - 4$. To show the upper bound we consider the special case of parallel rulings, whereas to show the lower bound we construct a family of polygon for which any ruling must induce a Reeb graph with $\Omega(n)$ leaves.  

\subsection{Upper Bound}
For any vector $v$, there is a parallel ruling $S$ defined by sweeping the line $\ell$ orthogonal to $v$ across $\R^2$.
We think of $v$ as the ``height'' direction, and the function \mbox{$f_{v}: P \rightarrow \R$} maps $x$ to its ``height'', $\langle v, x \rangle$, in the direction $v$. Here $\langle \cdot, \cdot \rangle$ denotes the dot product.
The Reeb graph $\Reeb(f_v)$ is the quotient space constructed by contracting the connected components of $P \cap \ell$ to single points as $\ell$ sweeps through $\R^2$ in the direction $v$. We denote by $b$ the number of branch (internal) nodes in $\Reeb(f_v)$ and by $l$ the number of leaves.

A \emph{reflex vertex} $p$ of a polygon $P$ is a vertex whose interior angle is strictly greater than $180^{\circ}$; see Figure \ref{fig:cone}. We denote by $R(P)$ the set of reflex vertices of $P$ and define $k = |R(P)|$. Note that a reflex vertex cannot be on the convex hull of $P$. The reflex vertices play an important role in determining the number of leaves in a Reeb graph induced by a parallel ruling.
Each reflex vertex $p$ bears witness to a closed set of vectors which, if the rulings are induced by any vector $v$ in the set, eliminate $p$ from the Reeb graph. That is $p$ is not a critical point of $f_{v}$ and does not correspond to a node of $\Reeb(f_{v})$.

Let $p \in R(P)$ be a reflex vertex of $P$. Denote by $n_1$ and $n_2$ the normals to the edges $e_1$ and $e_2$ adjacent to $p$, respectively. We will consider two double cones defined by the normals $n_1, n_2$ at apex $p$. Consider the vector $\overline{n} = (n_1+n_2)/||(n_1 + n_2)||$ and notice that $\angle(\overline{n}, n_1) = \angle(\overline{n}, n_2)$. Define the closed double cone $C_{p} = \{v \in \R^2: \angle(\overline{n}, v) \geq \angle(\overline{n}, n_1)\}$ ; see Figure \ref{fig:cone}.

\begin{figure}[h!]
\begin{center}
\includegraphics[width=0.3\textwidth]{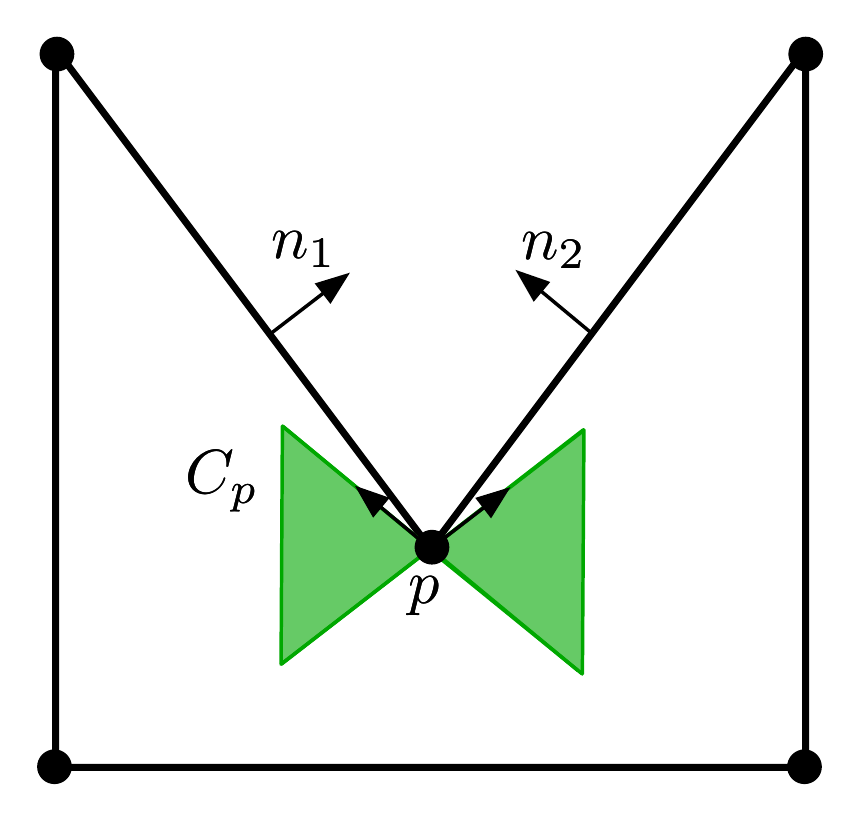}
\caption{A polygon with a reflex vertex at $p$. The cone $C_p$ is shown in green and is defined by the lines with directions $n_1$ and $n_2$ and intersecting $p$}
\label{fig:cone}
\end{center}
\end{figure}

\begin{lemma}
\label{lem:leafcond}
Let $P$ be a simple polygon with or without holes and $v$ be an arbitrary direction. Consider an arbitrary leaf node $q \in \Reeb(f_v)$. Then for all $p \in R(P)$, $p \not \in f_v^{-1}(q)$.
\end{lemma}
\begin{proof}
Suppose, for the sake of contradiction, that there exists a reflex vertex $p \in f_{v}^{-1}(q)$. Consider the ruling line $\ell$ at $p$, and divide $\ell$ into two rays $r_1, r_2$ with base point $p$. Since $q$ is a leaf node, the ruling line $\ell$ at $p$ locally intersects the interior of $P$ in a single connected component. Thus one of the two rays, say $r_1$, points into the exterior of $P$ between the two edges adjacent to $p$. However, this implies that we can perturb $\ell$ in directions $v$ and $-v$ while still locally intersecting the interior of $P$. This contradicts that fact that $q$ is a leaf node.
\end{proof}

\begin{lemma}
\label{lem:conecond}
Let $P$ be a simple polygon with or without holes, $p \in R(P)$ be a reflex vertex, and $v$ be an arbitrary direction. Then $p$ creates a branch node in $\Reeb(f_v)$ if and only if $v \not \in C_p$.
\end{lemma}
\begin{proof}
Note that, as a consequence of Lemma \ref{lem:leafcond}, $p$ can only create a branch node in $\Reeb(f_v)$. Suppose that $p$ creates a branch node in $\Reeb(f_v)$. For this to occur, the ruling line $\ell$ at $p$ locally intersects the interior of $P$ in two connected components. This happens if and only if the vector orthogonal to $\ell$ is not in the set $C_p$.
\end{proof}

From Lemmas \ref{lem:leafcond} and \ref{lem:conecond}, we see that $C_{p}$ defines precisely the set of vectors $v$ such that the ruling $f_v$ eliminates $p$ from $\Reeb(f_v)$. In Section \ref{sec:alg}, we will use the cones $C_p$ to compute the Reeb complexity of a polygon when restricted to the set of parallel rulings.

When $f_v$ is Morse, every branch node of $\Reeb(f_v)$ has degree 3. In this case we have that $2|E| = \sum_{u \in \Reeb(f_v)} \operatorname{deg}(u) = 3b + l$, where $|E|$ is the total number of edges in $\Reeb(f_v)$. Since the holes are disjoint, each hole creates a cycle in $\Reeb(f_v)$ adding one edge to the total number of edges, giving $|E| = b + l - 1 + h$. Combining the expressions we get the relation $l = b + 2 - 2h$. Note that when $h = 0$ we recover the relationship between the number of internal nodes and the number of leaves in a tree. Furthermore when $f_v$ is not Morse, the equality becomes the inequality $l \geq b + 2 - 2h$.

\begin{lemma}
Let $P$ be a simple polygon with $h$ holes, $k$ be the number of reflex vertices, and $v$ be an arbitrary direction. If $f_v$ is Morse, then $\Reeb(f_v)$ has at most $k + 2 -2h$ leaves.
\end{lemma}
\begin{proof}
In the worst case, the vector $v$ is not in $C_p$ for any $p \in R(P)$. By Lemma \ref{lem:conecond} every reflex vertex creates a branch node in $\Reeb(f_v)$. Since $f_v$ is Morse, we have the relationship $l = b + 2 - 2h \leq k + 2 - 2h$. 
\end{proof}

This result is tight in that there exists a polygon $P$ and a direction $v$ such that $\Reeb(f_v)$ has exactly $k+2$ leaves; see Figure \ref{fig:reeb}. However it is easy to construct polygons where all but a constant number of vertices are reflex vertices. In such cases we can bound the number of branch nodes in the Reeb graph much more tightly.

\begin{figure}[h!]
\begin{center}
\includegraphics[width=0.4\textwidth]{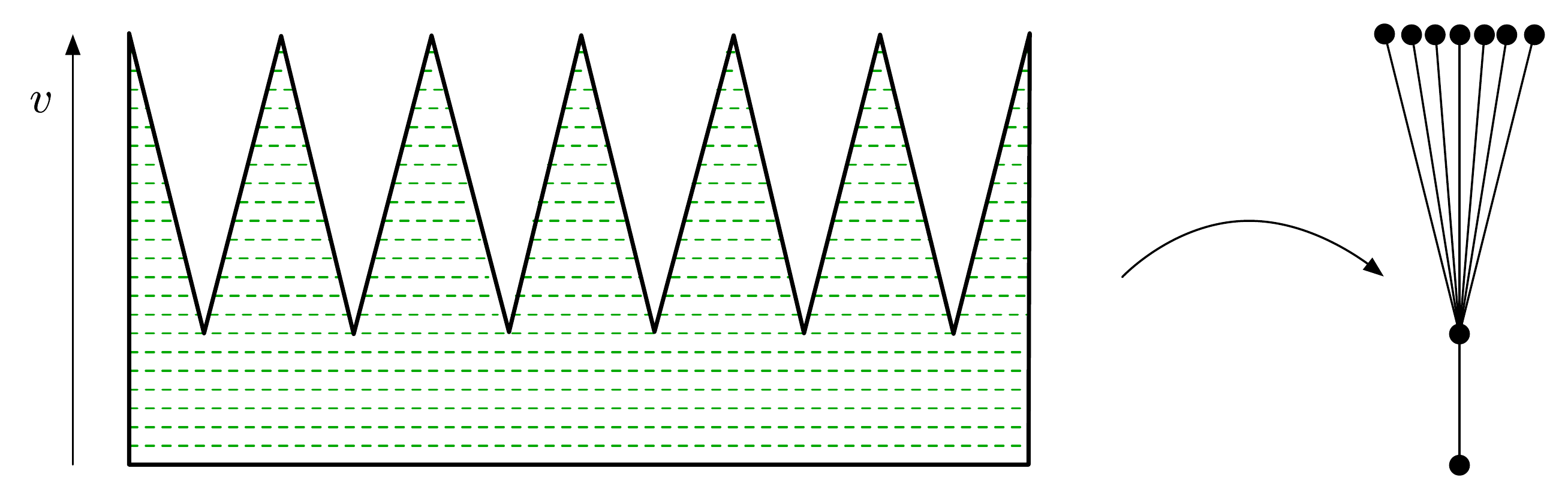}
\caption{A polygon with $6$ reflex vertices and a parallel ruling for which the Reeb graph has $8$ leaves. However had we taken $v$ to be the $x$-direction, the Reeb graph would only have $2$ leaves, and the ruling would be simple.}
\label{fig:reeb}
\end{center}
\end{figure}


\begin{lemma}
\label{lem:branch}
Let $P$ be a simple polygon with $h$ holes and $n$ vertices, and let $v$ be an arbitrary direction. Then number of branch nodes $b \leq \lfloor\frac{n}{2}-1 + h\rfloor$.
\end{lemma}
\begin{proof}
Let $m$ denote the number of non-reflex vertices and note that $n = k + m$. By Lemma \ref{lem:conecond}, a reflex vertex either forms a branch node or is eliminated from $\Reeb(f_v)$; it follows that $b \leq k$. Similarly, as a consequence of Lemma \ref{lem:leafcond}, every leaf node is created by the ruling passing over one or more non-reflex vertices and it follows that $l \leq m$. Combining these inequalities with the inequality $l \geq b + 2 - 2h$, we have that 
\begin{align*}
b + 2 -2h &\leq l\\
           &\leq m\\
2b + 2 -2h  &\leq m + b\\
       &\leq m + k\\
       &= n\\
    b &\leq \frac{n - 2 + 2h}{2}.
\end{align*}
\end{proof}

When $f_v$ is Morse, we can use the relation $l = b + 2 - 2h$ to bound the Reeb complexity of any ruling as $\lfloor \frac{n}{2} +1 - h \rfloor \leq \lfloor \frac{n}{2} + 1 \rfloor$. Notice that this bound holds for any arbitrary ruling, not just parallel rulings, as increasing the set of rulings considered only decreases the Reeb complexity. 

\begin{theorem}
\label{thm:upper}
Let $P$ be a simple polygon with $h$ holes and $n$ vertices. Let $S$ be any ruling of $P$. Then the Reeb complexity is at most $\lfloor \frac{n}{2} + 1 \rfloor$. 
\end{theorem}

\subsection{Lower Bound}

Consider once again the example shown in Figure \ref{fig:reeb}. Had we chosen the direction orthogonal to $v$, the Reeb graph would have only 2 leaves. To establish a lower bound, and thus show that our result is asymptotically tight, we construct a family of polygons whose Reeb complexity is $\Omega(n)$.

Consider two concentric circles $C_1, C_2$ centered at the origin and with radii $r_1, r_2$ respectively, where $r_1 \gg r_2$. We parameterize $C_1(\theta) = r_1(\sin{\theta}, \cos{\theta})$ and $C_2(\phi) = r_2 (\sin{\phi}, \cos{\phi})$. For each $n$, we construct a set of $2n$ vertices, $n$ of which will be placed on $C_1$, with the remaining $n$ vertices being placed on $C_2$.
The first set of $n$ vertices are placed on $C_1$ at $\theta_i = \frac{2\pi i}{n}$ for $i \in [n-1]$. The second set of $n$ vertices are placed on $C_2$ at $\phi_i = \frac{(2i+1)\pi}{n}$ for $i \in [n-1]$. The edges of the polygon are constructed by connecting the $i$th vertex of $C_1$ to the $i-1$ and $i$th vertices of $C_2$. See Figure \ref{fig:lowerbound}. Notice that every vertex on $C_2$ is a reflex vertex of the polygon.

\begin{figure}[h!]
\begin{center}
\includegraphics[width=0.38\textwidth]{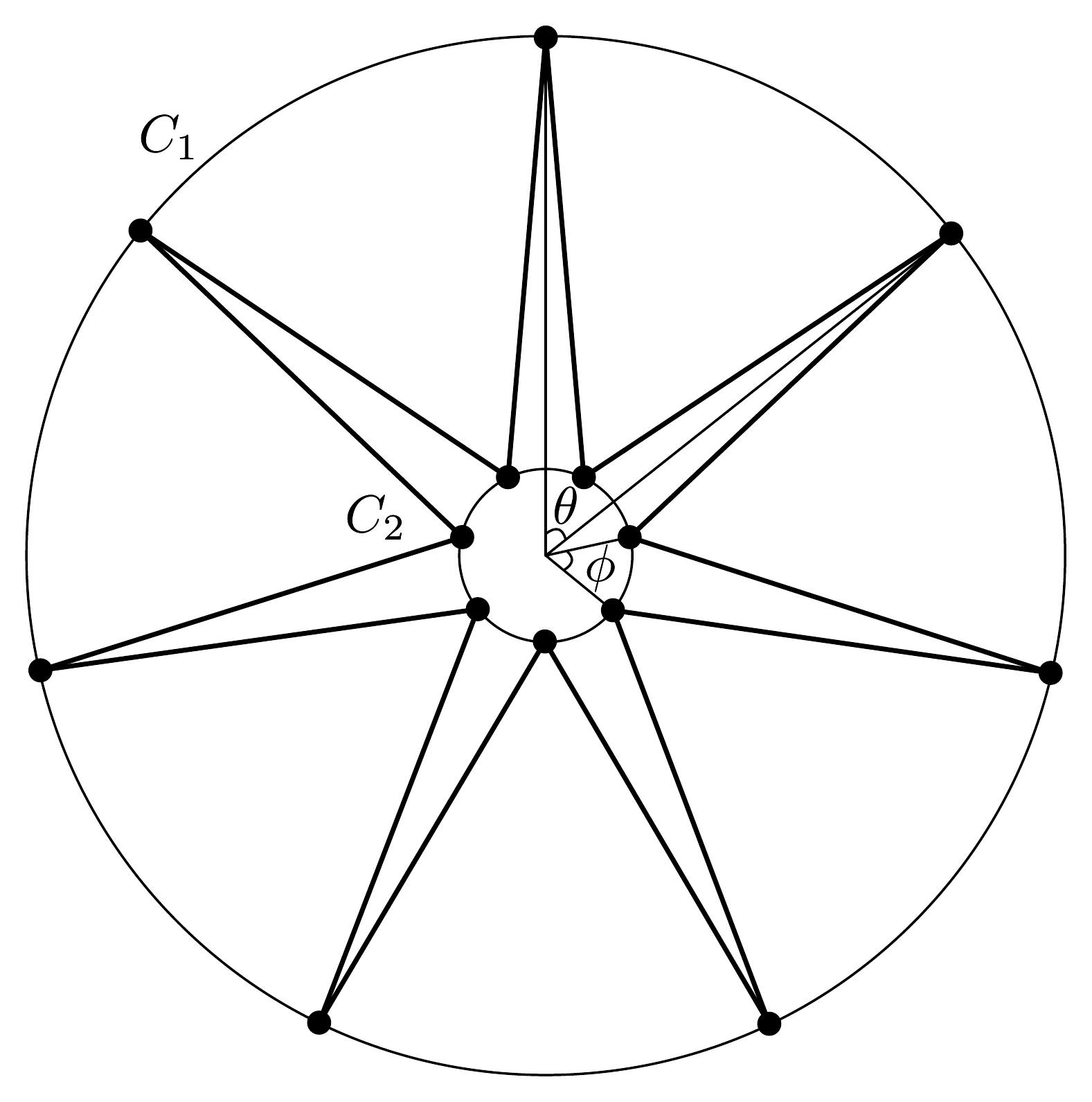}
\caption{Our lower bound construction for $n = 7$.}
\label{fig:lowerbound}
\end{center}
\end{figure}
Now consider a vertex $p$ on $C_1$, and suppose that some (not necessarily parallel) ruling $S$ eliminates $p$ from $\Reeb(S)$. Then there exists some line segment $s = (p,q)$ of the ruling $S$ with endpoint $p$. The other endpoint of $s$ can be contained in one of only two (when $n$ is odd) or three (when $n$ is even) other spikes of the polygon, due to the limited visibility at $p$. We prove this statement in the following paragraphs. Crucial to this argument is the fact that $q$ must be on the boundary of $P$ in a spike different than that of $p$ for $S$ to be a valid ruling.

Let $p_1, p_2$ be the vertices on $C_2$ that are adjacent to $p$. The length of the segment $|p_1p_2| \leq \frac{2\pi}{n}r_2$, the length of the arc connecting $p_1$ and $p_2$. The affine hulls of the edges $pp_1$ and $pp_2$ each intersect $C_2$ in two points. Let $p'_1$ and $p'_2$ be the intersections not equal to $p_1$ and $p_2$; see Figure \ref{fig:lowerbound2}. We will show that the length of the segment $|p'_1p'_2| \leq \frac{2\pi}{n} \frac{r_1 + 1}{r_1 - 1}$, which, for an appropriate choice of $r_1$, covers at most $2$ intervals when $n$ is odd and $3$ intervals when $n$ is even.

\begin{figure}[h!]
\begin{center}
\includegraphics[width=0.4\textwidth]{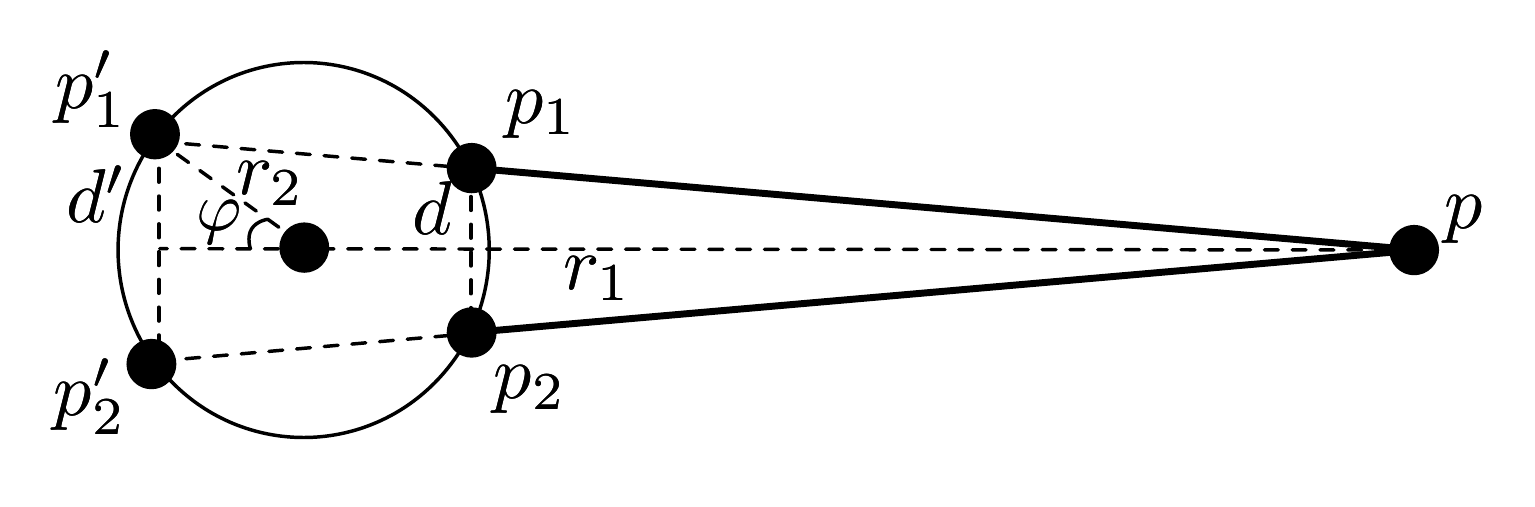}
\caption{For the spike at $p$ we consider the affine hulls of the edges adjacent to $p$. The intersection of these affine hulls with $C_2$ define $p'_1$ and $p'_2$. The point $q$ can lie in any of the spikes spanned by the segment $p'_1p'_2$.}
\label{fig:lowerbound2}
\end{center}
\end{figure}

Let $m$ and $m'$ be the midpoints of the segments $p_1p_2$ and $p'_1p'_2$ respectively. Notice that the triangles $\triangle p_1mp$ and $\triangle p'_1 m' p$ are similar. Define the lengths $d = |p_1m|$ and $d' = |p'_1m'|$. Since $\triangle p_1mp$ and $\triangle p'_1 m' p$ are similar we have the relationship $\frac{d'}{d} = \frac{|pm'|}{|pm|}$. We can write the length $|pm| = r_1 - r_2 + \delta$ for some $\delta > 0$. Here $\delta$ is the distance between $m$ and $C_2$. Similarly we can write $|pm'| = r_1 + r_2\sin{\varphi}$; see Figure \ref{fig:lowerbound2}. Then we have that
\begin{align*}
d' &= d \frac{|pm'|}{|pm|}\\
   &= d \frac{r_1 + r_2\sin{\varphi}}{r_1 - r_2 + \delta}\\
   &\leq \frac{\pi}{n} \frac{r_1 + r_2\sin{\varphi}}{r_1 - r_2 + \delta}\\
   &\leq \frac{\pi}{n} \frac{r_1 + r_2}{r_1 - r_2 + \delta}\\
   &\leq \frac{\pi}{n} \frac{r_1 + r_2}{r_1 - r_2},
\end{align*}
where the first inequality follows from $|p_1p_2| \leq \frac{2\pi}{n}r_2$, the second from $\sin{\varphi} \leq 1$, and the third from $\delta > 0$. Then taking $r_2 = 1$ we have that
\begin{align*}
|p'_1p'_2| &= 2d'\\
           &\leq \frac{2\pi}{n} \frac{r_1 + 1}{r_1 - 1}.
\end{align*}
The vertices on $C_2$ are spaced so that each interval has arc-length $\frac{2\pi}{n}$. All that remains is to compute the number of intervals a segment of length $\frac{2\pi}{n}\frac{r_1 + 1}{r_1 - 1}$ can cover. When $r_2 = 1$, this amounts to computing the angle $2\varphi$, shown in Figure \ref{fig:lowerbound2}, when $d'$ is at its maximum value. The angle $2 \varphi = 2\arcsin{\frac{d'}{r_2}} \leq 2\arcsin{\left(\frac{\pi}{n}\frac{r_1 + 1}{r_1 - 1}\right)}$. Then the maximum number of intervals spanned by the segment $|p'_1p'_2|$ is
\begin{equation*}
\frac{2\arcsin{\left(\frac{\pi}{n} \frac{r_1 + 1}{r_1 - 1}\right)}}{\frac{2\pi}{n}} = \frac{n}{\pi} \arcsin{\left(\frac{\pi}{n}\frac{r_1 + 1}{r_1 - 1}\right)},
\end{equation*}
which approaches $\frac{r_1 + 1}{r_1 - 1}$ as $n \rightarrow \infty$. For the purposes of this construction we take $r_1 = 4$.  When $r_1 = 4$, for all $n \geq 7$, the segment $|p'_1p'_2|$ can span at most $2$ intervals when $n$ is odd and at most $3$ when $n$ is even.

Suppose that $p$ corresponds to the $i$th spike. When $n$ is odd $q$ must intersect one of only two possible intervals, those corresponding to the spikes $i + \lfloor n/2 \rfloor$ and $i + \lceil n/2\rceil$. The $i$th and $i+1$th spike have overlapping regions of a single interval. However if a ruling line attempts to eliminate the vertex on the $i+2$th spike, the resulting line segment must intersect $s$. This is impossible in any valid ruling, and it follows that at most $4$ spikes can be eliminated from $\Reeb(S)$. Thus the Reeb complexity of this polygon is at least $n - 4$. Note that the polygon has $2n$ vertices in total, so this matches our upper bound. In conjunction with Theorem \ref{thm:upper}, we've established the following theorem.

\begin{theorem}
Let $P$ be a simple polygon with $h$ holes and with $n$ vertices. Let $S$ be any ruling of $P$. Then the Reeb complexity of $P$ is upper bounded by $\frac{n}{2}+1$. Furthermore there exists simple polygons for which the Reeb complexity is at least $\frac{n}{2}-4$. 
\end{theorem}

While our bound is asymptotically tight, the additive difference between the example used to establish the lower bound and the upper bound proved in Theorem \ref{thm:upper} is $5$. It remains open whether there exists a polygon $P$ for which every direction induces a Reeb graph with exactly $\frac{n}{2}+1$ leaves.
\section{Computing the Reeb Complexity for Parallel Rulings}
\label{sec:alg}
Given a simple polygon $P$ with $h$ holes and $n$ vertices we wish to compute the Reeb complexity of $P$. In Section \ref{sec:con}, we conjecture that this problem is NP-complete for general rulings. In the special case of parallel rulings, we show that the problem can be solved in $O(n\log{n})$ time.

By Lemma \ref{lem:conecond}, finding a parallel ruling of minimum Reeb complexity is equivalent to finding a vector $v$ that is contained in the maximum number of cones $C_p$. We use the standard duality transform that maps a point $(a, b)$ to the line $\ell = \{(x,y) : y = ax - b\}$. In the dual plane, a parallel ruling $S$ dualizes to a vertical line, because each line in $S$ has the same slope. Similarly, the two lines $\ell_p, \ell'_p$ that bound the cone $C_p$ dualize to two points $(m_p, c_p), (m'_p, c'_p)$ where $m_p, m'_p$ are the slopes of $\ell_p, \ell'_p$ and $-c_p, -c'_p$ are the $y$-intercepts.

\begin{figure}[h!]
\begin{center}
\includegraphics[width=0.45\textwidth]{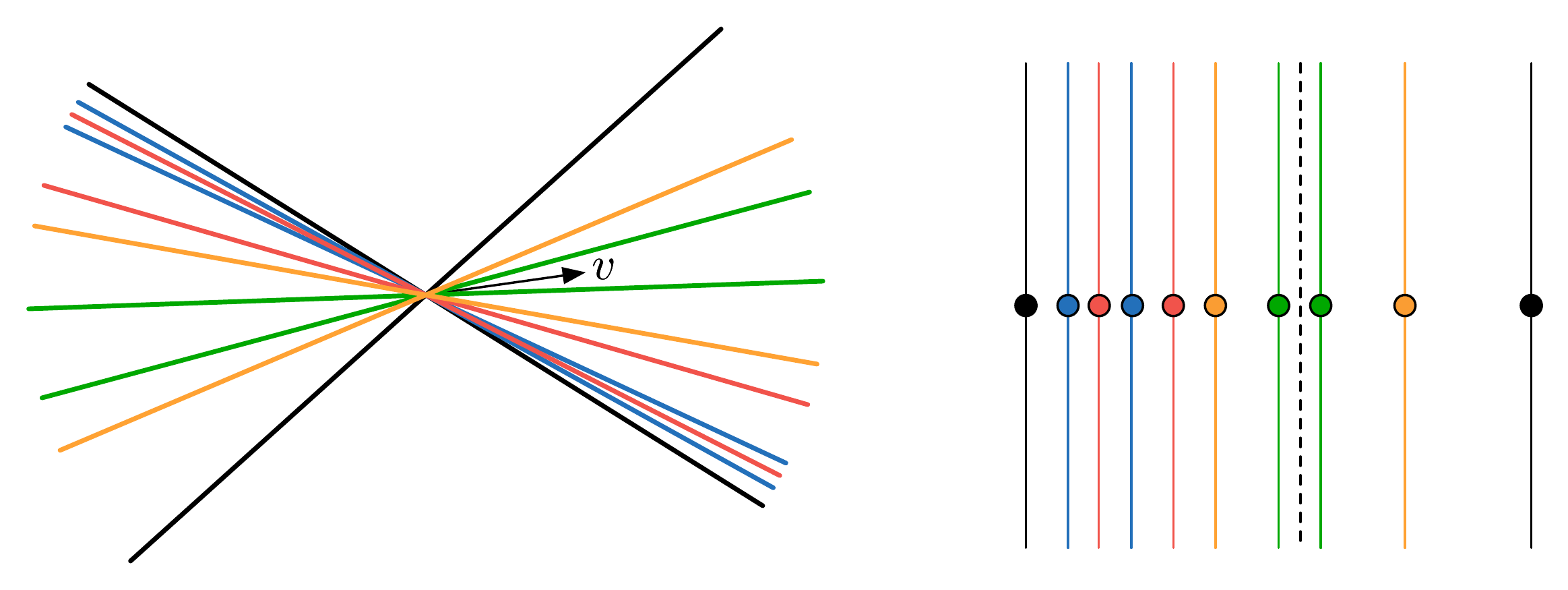}
\caption{The set of cones for each reflex vertex translated to the origin. The boundaries of these cones dualize to points in the dual plane. We disregard the $y$-coordinate of the dualized points and consider the resulting list of slopes on the $x$-axis. Then finding a vector in the maximum number of cones is equivalent to finding a line in the maximum number of intervals. Dualizing the entire dotted line as a set of points, gives the desired ruling.}
\label{fig:duality}
\end{center}
\end{figure}

The algorithm begins by computing the cone $C_p$ for each reflex vertex $p \in R(P)$. The duality transform is applied to the set of lines  $\cup_{p}\{\ell_{p}, \ell'_p\}$, giving the set of points $\cup_{p}\{(m_p, c_p), (m'_p, c'_p)\}$. Notice that the $y$-intercept values can be disregarded, as we are interested in a vector $v$ based at the origin that lies in the maximum number of cones translated to the origin. The set of slopes $I = \cup_{p}\{(m_p, m'_p)\}$ define a set of intervals. Sort the endpoints of the intervals and call the resulting list $L$. Finding a vector $v$ that is contained in the maximum number of cones is equivalent to finding the vertical line that lies in the maximum number of intervals in $I$.

There is one remaining caveat. Notice that traversing the list of slopes $L$ in increasing order corresponds, in the primal plane, to traversing the cones in rotary order starting with the vector $v = (0, -1)$ and performing a rotation of $180^{\circ}$. The vector $v$ may already lie in a subset of the cones $\mathcal{C}_v$. Consider $C_p \in \mathcal{C}_v$ and its corresponding pair $(m_p, m'_p)$. In this case, the first endpoint of the interval $(m_p, m'_p)$ that the traversal encounters in $L$ is an exit event, not an entry event. The pair $(m_p, m'_p)$ corresponds to the interval $(-\infty, m_p] \cup [m'_p, \infty)$. To account for this, the algorithm first computes $\mathcal{C}_v$, and labels the elements of $L$ with the correct entry/exit labels. The algorithm keeps a counter $c$, initialized with the value $|\mathcal{C}_v|$, and traverses $L$ incrementing $c$ on each entry event, and decrementing $c$ on each exit event. The maximum value of this counter $c_{\max}$ gives the minimum number of leaves $k - c_{\max} + 2 - 2h$, where $k$ is the number of reflex vertices.

The runtime of the algorithm is dominated by sorting $L$, which takes $O(n \log{n})$ time. The other steps of the algorithm -- computing the cones $C_p$, dualizing the boundary lines $\cup_{p}\{\ell_p, \ell'_p\}$, computing the subset $\mathcal{C}_v$, assigning the correct labels to the intervals, and computing $c_{\max}$ -- can all be done in $O(n)$ time. The correctness of the algorithm follows from Lemma \ref{lem:conecond}. We have established the following theorem.

\begin{theorem}
Let $P$ be a simple polygon with $h$ holes and $n$ vertices. The Reeb complexity of $P$, restricted to the set of parallel rulings, can be computed in $O(n \log{n})$ time.
\end{theorem}

Note that the algorithm presented in this section is equivalent to an algorithm in the primal space where a vector is rotated once around the origin. As the vector rotates around the origin, the algorithm keeps track of entry and exit events defined by each cone. We chose to present the algorithm in the dual space because future extensions to more general classes of rulings will likely operate in the dual space. As shown in Figure \ref{fig:duality}, a parallel ruling corresponds to a vertical line in the dual space, since each line segment of the ruling has identical slope. More general rulings correspond to curves in the dual space. Characterizing the set of curves that correspond to valid rulings of a polygon is likely to be an important first step to settling algorithmic questions related to Reeb complexity.
\section{Conclusions and Open Problems}
\label{sec:con}
Many problems on Reeb complexity of polygons and rulings remain open.
\begin{enumerate}
  \item Give an algorithm to determine if the Reeb complexity of a polygon is at most a given bound $b$.
  We conjecture that this problem is NP-Complete.
  \item A special case of the problem above is to test if a polygon admits a simple ruling.
  When this problem was posed at the open problem session of CCCG 2016, David Eppstein observed that a polygon admits a simple ruling if and only if some subdivision of the edges results in a polygon that admits a Hamiltonian triangulation~\cite{arkin94hamiltonian}.
  It may be possible to adapt the algorithm in that paper to this problem.
  There is likely also a connection to sweepable polygons~\cite{bose05generalizing}, 2-walkable~\cite{bhattacharya01optimal} polygons, and algorithms for detecting them.
  \item What rulings correspond to physically realizable rulings?  A similar problem is to characterize the rulings that result from a given support set under the effects of gravity.
  \item Is every Reeb graph of a ruling on $P$ also the Reeb graph of a continuous function on $P$?
  \item A ruling is called proper if no two line segments share an endpoint.  Can the Reeb complexity of a polygon change if we only permit proper rulings?
\end{enumerate}

\small
\bibliographystyle{abbrv}
\bibliography{paper}

\end{document}